\documentclass{article}
\usepackage[utf8]{inputenc}
\usepackage{amsmath,amsfonts,amssymb,mathrsfs}
\usepackage{amsthm}
\usepackage[margin=1in]{geometry}
\usepackage[noend]{algpseudocode}
\usepackage{setspace}
\usepackage{enumitem}
\usepackage{algorithm}
\usepackage{xcolor}
\usepackage{hyperref}
\definecolor{ForestGreen}{rgb}{0.1333,0.5451,0.1333}
\hypersetup{
    colorlinks=true,
    linkcolor=blue,
    citecolor=ForestGreen,
    filecolor=magenta,
    urlcolor=cyan
}
\usepackage{cleveref}

\usepackage{tikz}
\usetikzlibrary{external,calc,arrows,decorations.pathmorphing,positioning,shapes}
\usepackage{thmtools}

\theoremstyle{plain}
\newtheorem{theorem}{Theorem}
\newtheorem{lemma}[theorem]{Lemma}
\newtheorem{invariant}{Invariant}
\newtheorem{observation}{Observation}

\theoremstyle{definition}
\newtheorem{definition}{Definition}
\newtheorem{remark}{Remark}

\usepackage{mathtools}

\newcommand{\istrut}[2][0]{\rule[- #1 mm]{0mm}{#1 mm}\rule{0mm}{#2 mm}}
\newcommand{\rb}[2]{\raisebox{#1 mm}[0mm][0mm]{#2}}

\newcommand{\ignore}[1]{}
\newcommand{\ang}[1]{\left\langle #1 \right\rangle}
\newcommand{\floor}[1]{\left\lfloor #1 \right\rfloor}

\newcommand{\bydef}{\stackrel{\mathrm{def}}{=}}
\newcommand{\Var}{\mathbf{Var}}

\newcommand{\last}{\ensuremath{\operatorname{last}}}

\def\E{\ensuremath{\mathbf{E}}}
\def\BB{\ensuremath{\mathsf{BB}}}
\def\sgn{\ensuremath{\mathrm{sgn}}}

\def\CORR{\ensuremath{\operatorname{corr}}}
\def\DEV{\ensuremath{\operatorname{dev}}}

\def\OCORR{\ensuremath{\operatorname{\sigma-corr}}}

\newcommand{\XMAX}{\ensuremath{X_{\mathrm{max}}}}
\newcommand{\KMAX}{\ensuremath{K_{\mathrm{max}}}}

\newcommand{\BrachaAgreement}{\textsf{Bracha-Agreement}}
\newcommand{\CoinFlip}{\ensuremath{\textsf{Coin-Flip}}}
\newcommand{\ReliableBroadcast}{\textsf{Reliable-Broadcast}}

\newcommand{\IteratedBlackboard}{\textsf{Iterated-Blackboard}}
\newcommand{\WeightUpdate}{\ensuremath{\mathsf{Weight\text{-}Update}}}
\newcommand{\RisingTide}{\ensuremath{\mathsf{Rising\text{-}Tide}}}

\newcommand{\poly}{\operatorname{poly}}
\newcommand{\Akeep}{A_{\operatorname{keep}}}
\newcommand{\Aflip}{A_{\operatorname{flip}}}
\newcommand{\Adec}{A_{\operatorname{dec}}}
\newcommand{\bias}{\operatorname{bias}}
\newcommand{\biasbar}{\overline{\bias}}

\title{Byzantine Agreement with Optimal Resilience\\
via Statistical Fraud Detection}

\author{Shang-En Huang\\
University of Michigan
\and
Seth Pettie\\
University of Michigan
\and 
Leqi Zhu\\
University of Michigan}
\date{}

\begin{document}

\maketitle

\begin{abstract}
Since the mid-1980s it has been known that 
Byzantine Agreement can be solved with probability 1 asynchronously, 
even against an omniscient, computationally unbounded adversary that 
can adaptively \emph{corrupt} up to $f<n/3$ parties. 
Moreover, the problem is insoluble with $f\geq n/3$ corruptions.
However, Bracha's~\cite{Bracha1987} 1984 protocol (see also Ben-Or~\cite{Ben-Or83}) 
achieved $f<n/3$ resilience 
at the cost of \emph{exponential} expected latency $2^{\Theta(n)}$, 
a bound that has never been improved in this model with $f=\floor{(n-1)/3}$ corruptions.

In this paper we prove that Byzantine Agreement in the asynchronous, full information model can be solved
with probability 1 against
an adaptive adversary that can corrupt $f<n/3$ parties, while incurring only
\emph{polynomial latency with high probability}.  
Our protocol follows earlier 
polynomial latency protocols 
of King and Saia~\cite{KingS2016,KingS2018}
and 
Huang, Pettie, and Zhu~\cite{HuangPZ22}, 
which had \emph{suboptimal}
resilience, namely $f \approx n/10^9$~\cite{KingS2016,KingS2018}
and $f<n/4$~\cite{HuangPZ22}, respectively.

Resilience $f=(n-1)/3$ is uniquely difficult as this is the point at which the influence of the 
Byzantine and honest players are of roughly 
equal strength.
The core technical problem we solve is to 
design a 
collective coin-flipping protocol that \emph{eventually} 
lets us flip a coin with an unambiguous outcome.  
In the beginning
the influence of the Byzantine players is too powerful
to overcome and they can essentially 
fix the coin's behavior at will.
We guarantee that after just a polynomial number of executions
of the coin-flipping protocol, either
(a) the Byzantine players fail to 
fix the behavior of the coin 
(thereby ending the game)
or 
(b) we can ``blacklist'' players such that the blacklisting rate for 
Byzantine players is at least as large as the blacklisting rate for 
good players.  The blacklisting criterion is based on a 
simple statistical test of \emph{fraud detection}.
\end{abstract}

\section{Introduction}

In the Byzantine Agreement problem~\cite{PeaseSL80,LamportSP82}, 
$n$ players begin with input values 
in $\{-1,1\}$ and each must \emph{decide} 
an output value in $\{-1,1\}$ subject to:
\begin{description}
\item[Agreement.] All uncorrupted players must {\bf decide} 
the same value (and only that value).
\item[Validity.] If all uncorrupted players {\bf decide} $v$, 
then at least one such player had $v$ as its input.
\item[Termination.] Each uncorrupted player terminates 
the protocol with probability 1.
\end{description}
The difficulty of this problem depends on the \emph{strength} of the adversary and assumptions on the communication medium.  We consider a standard asynchronous model of communication against a \emph{strong} adversary.  
Each player can send point-to-point messages to other players, which can be \emph{delayed} arbitrarily by the adversary, but not dropped or forged.
In addition, the adversary is aware of the internal 
state of every player, is computationally unbounded, 
and may adaptively \emph{corrupt} 
up to $f$ players; these are also known as \emph{Byzantine} players.
Once corrupted, the behavior of 
a player is arbitrary, 
and assumed to be controlled by the adversary.
Following Ben-Or~\cite{Ben-Or83}, Bracha~\cite{Bracha1987} proved that Byzantine Agreement can be solved in this model when 
$f<n/3$.  
The protocols of Ben-Or and Bracha 
are not entirely satisfactory because 
they have latency \emph{exponential} in $n$.
(In the asynchronous model, 
a protocol has \emph{latency} $L$ if, in 
an execution in which every 
message delay is bounded by $\Delta$ and all 
local computation is instantaneous, each player
halts by time $L\Delta$.\footnote{Note that $\Delta$ is unknown to the players and cannot be used to detect crash-failures.  Moreover, since there is no \emph{minimum} message delay, these conditions 
place no constraints on message delivery order.})

The Byzantine Agreement problem 
has been solved satisfactorily
in stronger communication 
models or against a weaker adversary
than the ones we assume.

\paragraph{Synchronized Communication.}
Lamport et al.~\cite{LamportSP82} proved that if communication occurs in \emph{synchronized} rounds, Byzantine Agreement can be obtained deterministically in $f+1$ rounds, where $f<n/3$.
Fischer and Lynch proved that round complexity $f+1$ is optimal~\cite{FischerL82}.  
The communication complexity 
of~\cite{LamportSP82} is 
exponential, and was later reduced to 
polynomial by Garay and Moses~\cite{GarayM98}.
Dwork, Lynch, and Stockmeyer~\cite{dwork1988consensus} developed agreement 
protocols under weakly synchronous models.

\paragraph{Impossibility Results.} 
Fischer, Lynch, and Patterson~\cite{FischerLP85} proved that the problem cannot be solved deterministically, in an \emph{asynchronous} system in which just \emph{one} player is 
subject to a crash failure.
This result is commonly known as \emph{FLP Impossibility}.
Thus, to solve Byzantine Agreement we must assume some level of synchronization
\emph{or} randomization.  
Even with randomization, the problem is insoluble
when $f\geq n/3$.  The proof of this result is straightforward in the asynchronous
model~\cite[Thm.~3]{BrachaT85} and 
more complicated in the 
synchronized model~\cite{PeaseSL80,FischerLM86}.

\paragraph{Cryptographic Solutions.}
Against a computationally bounded adversary cryptography becomes useful.  
Byzantine Agreement can be solved against such an adversary 
in $O(1)$ latency against an adversary controlling  $f<n/3$ 
players\footnote{This assumes that RSA encryption cannot be broken by a polynomially bounded adversary.}~\cite{CachinKS05,CanettiR93}; see also~\cite{BermanG93,FeldmanM97}.  

\paragraph{Non-Adaptive Adversaries.} 
The ability to \emph{adaptively} corrupt players is surprisingly powerful.  
Goldwasser, Pavlov, and Vaikuntanathan~\cite{GoldwasserPV06}, improving~\cite{ben2006byzantine},
considered a \emph{synchronized}, full information model in which the adversary corrupts 
up to $f$ players up front, i.e., it is \emph{non-adaptive}.  
They proved that Byzantine Agreement can be solved with resiliency $f<n/(3+\epsilon)$ in $O(\log n/\epsilon^2)$ rounds.  
See Chor and Coan~\cite{chor1985simple} for prior results in similar adversarial models~\cite{chor1985simple}.

Kapron et al.~\cite{KapronKKSS10} developed a Byzantine Agreement protocol 
in the \emph{asynchronous}, full information model, in which the 
adversary corrupts $f<n/(3+\epsilon)$ players non-adaptively.  
Their protocol has different parameterizations, and can achieve agreement 
in quasipolynomial latency with probability $1-1/\poly(n)$, 
or polylogarithmic latency with probability $1-1/\poly(\log n)$. 
When these protocols err, they do not satisfy the {\bf Agreement} and {\bf Termination} 
criteria, and may deadlock or terminate without agreement.

\paragraph{Limits of Fully Symmetric Protocols.} 
Lewko~\cite{Lewko11} proved that in the asynchronous, 
full information model, a certain class of ``fully symmetric'' Byzantine Agreement protocols has latency $2^{\Omega(n)}$ when $f=\Theta(n)$.
This class was designed to capture 
Ben-Or~\cite{Ben-Or83} and Bracha~\cite{Bracha1987}
but is broader.  Protocols in this class make state transitions
that depend on the \emph{set} of validated messages received, 
but may not take into account the transaction history of the sender.
In retrospect, Lewko's result can be seen as justifying two strikingly different approaches for improving~\cite{Ben-Or83,Bracha1987} in the asynchronous, full information model.  
The first is to break symmetry by having the 
players take on different roles:
this is necessary to implement 
Feige's \emph{lightest bin} rule and other routines
in Kapron et al.'s~\cite{KapronKKSS10} protocol.
The second is to stay broadly within 
the Ben-Or--Bracha framework, 
but periodically \emph{blacklist} players 
after accumulating statistical evidence of 
fraudulent coin flips in their transaction history.
This is the approach taken by King and Saia~\cite{KingS2016,KingS2018}
and Huang et al.~\cite{HuangPZ22}.

\paragraph{Fraud Detection.}
King and Saia~\cite{KingS2016,KingS2018} 
presented two Byzantine Agreement protocols with polynomial latency.  
The first uses \emph{exponential} local computation and is 
resilient to $f<n/400$ adaptive corruptions.  
The second uses polynomial local computation and is 
resilient to $n/(0.87 \times 10^9)$ adaptive corruptions.  
Huang, Pettie, and Zhu~\cite{HuangPZ22} 
recently proposed a different fraud detection mechanism,
which lead to a Byzantine Agreement protocol with 
polynomial latency and polynomial local 
computation that is
resilient to $f<n/4$ adaptive corruptions.  
However, the specific 
statistical tests 
used by Huang et al.~\cite{HuangPZ22} 
are incapable of closing the gap 
from $f<n/4$ to $f<n/3$.  
See \autoref{sect:prelim}.

\subsection{New Results}

One feature of the asynchronous model is that every 
player must perpetually entertain the possibility 
that $f$ players have \emph{crashed} and will never be heard from again.  
Thus, when $n=3f+1$, the number of fully participating players at any stage is $n-f=2f+1$ and up to $f$ of these players may be corrupt!  Among the set of \emph{participating} players, 
the good players hold the thinnest 
possible majority: $f+1$ vs.~$f$.

We develop a special coin-flipping protocol to be used in Bracha's framework~\cite{Bracha1987,Ben-Or83}
when the corrupt and non-corrupt players have roughly 
equal influence.
Initially all players have weight $1$.
The coin-flipping protocol has the property that if the corrupt players repeatedly foil its attempts to flip a global coin, 
then we can fractionally \emph{blacklist} players 
(reduce their weights) in such a way that the blacklisting rate for 
good players is only infinitesimally larger than the blacklisting rate for corrupt players. 
Specifically, we guarantee that among any $n-f=2f+1$ participating players, 
the total weight of the good players minus the total weight of the corrupt players 
is bounded away from zero.  Eventually all corrupt players have their weights reduced to zero 
(meaning they have no influence over 
the global coin protocol) 
and at this point, the \emph{scheduling}
power of the adversary is insufficient to 
fix the behavior of the global coin.  
Agreement is reached in a few more 
iterations of Bracha's algorithm, with 
high probability.

The final result is a randomized $f$-resilient 
Byzantine Agreement protocol with latency 
(expected and with high probability) $\tilde{O}(n^4\epsilon^{-8})$,
where $n=(3+\epsilon)f$, $\epsilon \geq 1/f$.  
In other words, the latency ranges between 
$n^4$ and $n^{12}$, depending on $\epsilon$.
This latency-resiliency tradeoff is always
at least as good as~\cite{HuangPZ22},
but is slower than~\cite{KingS2016,KingS2018}
when $f<n/400$ or $f<n/(0.87 \times 10^9)$ is sufficiently 
small; see Table~\ref{tab:BA}.

\begin{table} [H]
\begin{tabular}{l|l|l}
\multicolumn{1}{l}{\textsf{Citation}} 
& 
\multicolumn{1}{l}{\textsf{Resilience}} 
& 
\multicolumn{1}{l}{\textsf{Latency \hspace{1cm} Local Computation per Message}}\\\hline\hline
\rb{-3}{Ben-Or~\cite{Ben-Or83}}  & $f<n/5$\istrut[2]{4}       & $2^{\Theta(n)}$\hfill $\poly(n)$\\\cline{2-3}
                        & $f=O(\sqrt{tn})$\istrut[2]{4}  & $\exp(t^2)$ \hfill $\poly(n)$\\\hline
Bracha~\cite{Bracha1987}& $f<n/3$\istrut[2]{4}       & $2^{\Theta(n)}$\hfill $\poly(n)$\\\hline
\rb{-3}{King \& Saia~\cite{KingS2016,KingS2018}}
                        & $f<n/400$\istrut[2]{4}     & $\tilde{O}(n^{5/2})$ \hfill $\exp(n)$\\\cline{2-3}
                        & $f<n/(0.87 \times 10^9)$\istrut[2]{4} & $O(n^3)$ \hfill $\poly(n)$\\\hline
Huang, Pettie \& Zhu~\cite{HuangPZ22}    & $f<n/(4+\epsilon)$\istrut[2]{4} & $\tilde{O}(n^4/\epsilon^4)$ \hfill $\poly(n)$\\\hline
\textbf{new}            & $f<n/(3+\epsilon)$\istrut[2]{4}       & $\tilde{O}(n^{4}/\epsilon^{8})$ \hfill $\poly(n)$\\\hline\hline
\end{tabular}
\caption{\label{tab:BA}
Randomized Byzantine Agreement protocols in the asynchronous, full information model 
against an adaptive adversary.  Here $\epsilon = \Omega(1/n)$.}
\end{table}

\subsection{Organization of the Paper}

In \autoref{sect:prelim} we review Bracha's algorithm~\cite{Bracha1987},
the coin-flipping protocols of King \& Saia and Huang et al.~\cite{KingS2016,HuangPZ22},
and the fraud detection mechanism of Huang et al.~\cite{HuangPZ22}.
We then walk through a few failed attempts to 
improve its resiliency to $f<n/3$.
The specific structure of these failures motivates several design choices
in our final protocol, many of which will not make sense without 
understanding the vulnerabilities of alternate solutions.
\autoref{sect:protocol} presents 
the protocol and its analysis.
We conclude in \autoref{sect:conclusion} 
with some remarks and open problems.

\section{Preliminaries, Misfires, and Dead Ends}\label{sect:prelim}

Bracha's protocol~\cite{Bracha1987} is based on Ben-Or's protocol~\cite{Ben-Or83}.
It improves the resiliency of~\cite{Ben-Or83} from $f<n/5$ to $f<n/3$ by 
introducing two mechanisms that constrain the misbehavior of corrupt players.
\begin{description}
\item[Reliable Broadcast.] The point-to-point communication network can be used to implement a simple \ReliableBroadcast{} primitive~\cite{Bracha1987} resilient to $f<n/3$ corruptions.  
It guarantees that if any good player attempts to broadcast a value $v$, then every good player 
eventually \emph{accepts} $v$ and only $v$.  
Moreover, if any corrupt player initiates a broadcast, 
then either all good players accept the same value $v$, and only $v$, or all good players accept nothing.
See~\cite{Bracha1987} for details of this primitive.
\item[Validation.]  The \ReliableBroadcast{} primitive allows us to assume that all relevant communication is public, 
via broadcasts.  Fix any protocol $\mathcal{P}$ based on broadcasts.  
Informally, a player $p$ \emph{validates} a message $m$ originating from $q$ if $p$ has already accepted and validated
a set of broadcasts that, were they to be received by $q$, 
would have caused $q$ to make a 
suitable state transition according 
to $\mathcal{P}$ and broadcast $m$.
See~\cite{Bracha1987} 
for details of validation.
\end{description}

The reliable broadcast primitive prevents the 
adversary from sending 
conflicting messages to 
different players, 
or convincing one player
to accept a broadcast
and another not to.
The validation mechanism prevents it from making state transitions logically inconsistent with the protocol $\mathcal{P}$.
Note, however, that in general $\mathcal{P}$ is \emph{probabilistic} and validation permits a series of transitions
that are logically possible but statistically unlikely.  
In summary, the adversary is 
characterized by the following powers.
\begin{description}
    \item[Full Information \& Scheduling.] The adversary knows the internal state of all players and controls the order in which messages are delivered.  It may delay messages arbitrarily.
    \item[Corruption \& Coin Flipping.] The adversary may \emph{adaptively} corrupt up to $f$ players as the execution of 
    the protocol progresses.  Once corrupted, a player \underline{\emph{continues to follow protocol}}, 
    except the adversary now chooses the outcomes of all of its coin flips.
\end{description}

\begin{algorithm}[H]
    \caption{\BrachaAgreement() \ \emph{from the perspective of player $p$}}
\begin{algorithmic}[1]
\Require{$v_p\in \{-1,1\}$.}
\Loop
\State
\ReliableBroadcast{} $v_{p}$ and \textbf{wait} until $n-f$ messages are validated from some set of players $S_p$.
{\par \enspace}
set $v_p \gets \sgn(\sum_{q\in S_p} v_q)$. \label{line:Bracha1}
\Comment{$\sgn(x)=1$ if $x\ge 0$ and $-1$ otherwise.}

\State
\ReliableBroadcast{} $v_{p}$ and \textbf{wait} until $n-f$ messages are validated.  
{\par \enspace}
\textbf{if} more than $n/2$ messages have some value $v^*$ \textbf{then} set $v_p \gets v^*$,
\textbf{otherwise} set $v_p \gets \bot$.
\label{line:Bracha2}

\State
\ReliableBroadcast{} $v_{p}$ and \textbf{wait} until $n-f$ messages 
are validated. 
{\par \enspace}
let $x_p$ be the number of $v^* \neq \bot$ messages validated by $p$.
\label{line:Bracha3}

\If{$x_p \geq 1$}\label{line:firstif}
    \State set $v_p \gets v^*$. \label{line:fplus1}
\EndIf

\If{$x_p \geq f+1$}\label{line:secondif}
    \State \textbf{decide} $v^*$. \label{line:decide-v}
\EndIf

\If{$x_p=0$} \label{line:thirdif}
    \State $v_p \gets \CoinFlip()$. \label{line:coinflip}
\Comment{Returns value in $\{-1,1\}$.}
\EndIf
\EndLoop
\end{algorithmic}
\end{algorithm}

The \BrachaAgreement{} protocol loops until all players \textbf{decide} a value.
As we will see, if any non-corrupt player \textbf{decides} in iteration $i$, then all non-corrupt players will \textbf{decide} by iteration $i+1$.
Suppose that, at Line~\ref{line:Bracha1}, 
a supermajority of
at least $(n+f+1)/2$ players 
hold the same value $v^*$.  
It follows that in any set $S_p$ of $n-f$ messages, 
$v^*$ will be in the majority, 
and all players will be forced to adopt $v^*$.\footnote{A corrupt player may successfully
broadcast the message $-v^*$ in Line~\ref{line:Bracha2} of \BrachaAgreement,
but no good player will \emph{validate} this 
message as it cannot be justified by any $n-f$ 
messages broadcast in Line~\ref{line:Bracha1}.} 
Thus, 
after Line~\ref{line:Bracha2}, all players 
will see $n-f > n/2$ messages containing $v^*$ and again broadcast $v^*$ in Line~\ref{line:Bracha3}, setting $x_p = n-f \geq f+1$ and \textbf{deciding} $v^*$ in Line~\ref{line:decide-v}.
In other words, if the adversary is to prolong the execution of this protocol,
it must avoid supermajorities of $(n+f+1)/2$ or more in every iteration.

Suppose the populations holding $-1$ and $1$ at Line~\ref{line:Bracha1} are 
more balanced. 
The scheduling power of the adversary is sufficient to fix $v_p$ arbitrarily 
in Lines~\ref{line:Bracha1} and \ref{line:Bracha2}.  
Note, however, that regardless of how the messages are delivered,
it is always the case that
$|x_p-x_q|\leq f$ for any $p,q$.  
Thus, if $\Akeep,\Adec,\Aflip$ are the populations with $x_p \in [1,f]$, $x_p\in [f+1,n-f]$,
and $x_p=0$, respectively, then either 
$\Aflip=\emptyset$ or $\Adec=\emptyset$.  
In the former case, a supermajority of players
hold the majority value $v^*$ 
and will \textbf{decide} 
in the current or following iteration.
In the latter case, 
the $\Akeep$ population keeps the value 
$v^*$ and the remaining population $\Aflip$
chooses their new value on the basis of a 
coin flip (Line~\ref{line:coinflip}).  
The key insight of \cite{Ben-Or83,Bracha1987} is that if 
\emph{all} good players in the second population $\Aflip$ 
happen to have $\CoinFlip$ return
$v^*$, then at least $n-f \geq (n+f+1)/2$ players will
go into the next iteration with the same value, 
and, according to the argument above, \textbf{decide} $v^*$ in that iteration. 

The probability that private coin-flipping creates a sufficiently large supermajority by chance is $2^{-\Theta(n)}$ when $f=\Theta(n)$.
In expectation the number of trials is $2^{\Theta(n)}$, 
hence the exponential expected latency 
of~\cite{Ben-Or83,Bracha1987}.
The protocols of Rabin~\cite{Rabin83} and Toueg~\cite{Toueg84} operate in a
similar fashion, but assume away the difficulty by supposing there is some mechanism to flip a shared coin, 
or reliably distribute shared randomness to players 
before the protocol begins.

\medskip 

The protocols of King and Saia~\cite{KingS2016,KingS2018} and Huang et al.~\cite{HuangPZ22} follow \BrachaAgreement, 
but implement $\CoinFlip$ (Line~\ref{line:coinflip}) 
as a ``global'' coin, 
which aims for two desirable, but ultimately unattainable guarantees:
\begin{description}
\item[Property (i)] all players agree on the \emph{same} value returned by $\CoinFlip()$, and 
\item[Property (ii)] the output of $\CoinFlip()$ is close to unbiased.
\end{description}
The problem of flipping a bounded-bias coin against adversarial 
manipulation is well studied.  
The problem can be solved against surprisingly large coalitions of corrupt players~\cite{Ben-OrL85,AjtaiL93,Saks89,AlonN93,RussellSZ02,HaitnerK20,KahnKL88,Feige99,BoppanaN00}.  This body of work assumes reliable communication (no dropped or delayed messages) and reliable computation (no crash failures).
Aspnes~\cite{Aspnes98} gave a lower bound that models aspects of an adaptive adversary in an asynchronous network.  In his coin-flipping game,
a vector of values $(v_1,\ldots,v_N)$ is generated as follows.  
Once $(v_1,\ldots,v_{i-1})$ are known, a random value $v_i'$ is 
generated\footnote{The distribution and range of $v_i'$ are arbitrary,
and may depend on $(v_1,\ldots,v_{i-1})$.} and the adversary may 
set $v_i\gets v_i'$ or \emph{suppress} it, setting $v_i\gets \bot$.
The outcome of the coin flip is some function $g(v_1,\ldots,v_{N})\in\{-1,1\}$.
If the adversary can suppress $t$ values, then $N=\Omega(t^2)$ for $g$ to have constant bias 
and $N =\Omega(t^2/\log^2 t)$ if the probability that $g=1$ and $-1$ are 
at least $1/\poly(t)$.  Aspnes~\cite{Aspnes98} proved that this result implies $\tilde{\Omega}(n)$ latency lower bounds on Byzantine Agreement in the asynchronous model, which was improved to $\Omega(n)$ by Attiya and Censor-Hillel~\cite{attiya2008tight}.  The moral of~\cite{Aspnes98,attiya2008tight} and related lower bounds against \emph{adaptive} adversaries, such as Haitner and Karidi-Heller's~\cite{HaitnerK20}, 
is that the aggregation function $g$ that implements
majority voting is at least close to optimal.
However, 
the Byzantine Agreement protocols against \emph{non-adaptive} adversaries,
such as~\cite{ben2006byzantine,GoldwasserPV06,KapronKKSS10} can afford to implement clever coin-flipping protocols that are not based exclusively on majority voting~\cite{Feige99,RussellSZ02}. 

\medskip 

The coin-flipping protocols of King and Saia~\cite{KingS2016} and Huang et al.~\cite{HuangPZ22} 
do not attempt to guarantee Properties (i) and (ii) immediately.  
Rather, after a sufficiently large number of invocations of $\CoinFlip$, if the adversary foils Properties (i,ii), 
it will leave behind enough statistical evidence that proves incriminating, allowing us 
to \emph{blacklist} suspicious players,
removing their explicit influence over subsequent calls to $\CoinFlip$.
When all corrupt players are blacklisted, 
the adversary still has the power of scheduling, 
but this power is insufficient to significantly 
delay agreement. 

The basis of King and Saia's~\cite{KingS2016} implementation of \CoinFlip{} is a shared \emph{blackboard} primitive, which was improved by Kimmett~\cite{Kimmett2020}, and again by Huang et al.~\cite{HuangPZ22}.   
A blackboard is an $m\times n$ matrix $\BB$, initially all blank ($\bot$).  The goal is to have each player $i$ write $m$ values successively to column $i$ (via \ReliableBroadcast{}s), 
and once the blackboard is full, to have all players agree on its contents.  Because up to $f$ players may crash, a \emph{full} blackboard is one in which $n-f$ columns have $m$ writes, and the remaining $f$ columns may be \emph{partial}.  
Due to the scheduling power of the adversary, every player $p$ 
sees a slightly different version $\BB^{(p)}$ of the ``true'' blackboard $\BB$, which is derived by replacing with $\bot$ 
the last write in some of the $f$ partial columns. Thus, $\BB^{(p)}$ and $\BB^{(q)}$ differ in at most $f$ entries.
In~\cite{KingS2016,HuangPZ22}
the $\CoinFlip$ routine is implemented as follows: 
every write to $\BB$ is a value in $\{-1,1\}$ chosen uniformly at random.
When $p$ finishes participating in the construction of $\BB$ it has
a view $\BB^{(p)}$ and sets the output of $\CoinFlip$ to be
$\sgn(\Sigma^{(p)})$, where $\Sigma^{(p)} = \sum_{j,q} \BB^{(p)}(j,q))$,
treating $\bot$s as zero.
Note that whenever $\Sigma^{(p)} \not\in [-f,f]$, 
$p$ can be sure that $\CoinFlip$ generates the same output for all players, even corrupt ones.\footnote{In particular,
in Line~\ref{line:Bracha1} 
of \BrachaAgreement,
if a message $v_q$ broadcast 
from $q$ is purportedly the
output of the last iteration's call 
to $\CoinFlip()$, player $p$ will validate
it only if it is \emph{possible} that 
$\sgn(\Sigma^{(q)}) = v_q$.}

We can of course execute the blackboard primitive iteratively~\cite{KingS2016}, 
but two players may disagree on the contents of \emph{each} blackboard in up to $f$ cells.  The \IteratedBlackboard{} protocol of Huang et al.~\cite{HuangPZ22} guarantees a stronger form of agreement. 
An iterated blackboard is an endless sequence $(\BB_1,\BB_2,\ldots)$ of blackboards, where $\BB_t$
is an $m(t)\times n$ matrix.  
After player $p$ finishes participating in 
the construction of $\BB_t$, it has a view $\BB^{(p,t)} = (\BB_1^{(p,t)},\ldots,\BB_t^{(p,t)})$ of the first $t$ blackboards.  It is guaranteed that $\BB_t^{(p,t)}$ and $\BB_t^{(q,t)}$ differ in at most $f$ cells in partial columns; it is also guaranteed that $\BB^{(p,t)}$ and $\BB^{(q,t)}$ differ in at most $f$ cells \emph{in total}, over all $t$ blackboards.  
In order to make this type of guarantee, during the construction of $\BB_{t+1}$, $p$ may record \emph{retroactive updates} to an earlier $\BB_{t'}$, $t'\leq t$, so that $\BB_{t'}^{(p,t+1)}$ 
records some writes 
to cells that were still $\bot$ in $\BB_{t'}^{(p,t)}$.

\begin{theorem}\label{thm:blackboard}
There is a protocol for $n$ players 
to generate an iterated 
blackboard $\BB$ that is resilient to $f<n/3$ Byzantine failures.
For $t \geq 1$, the following properties hold:
\begin{enumerate} 
    \item \label{thm:blackboard-full-partial} 
    Upon completion of the matrix $\BB_t$, each column consists of a 
    prefix of non-$\bot$ values and a suffix of all-$\bot$ values. 
    Let $\last(i) = (t',r)$ be the position of the last value written by
    player $i$, i.e., 
    $\BB_{t'}(r,i)\neq \bot$ and if $t'<t$ then $i$ has not 
    written to any cells of $\BB_t$.
    When $\BB_{t}$ is complete, 
    it has at least $n-f$ full columns and up to
    $f$ partial columns, i.e., $\last(i)\geq (t,m(t))$ for at least $n-f$ values of $i$.
    \item \label{thm:blackboard-ambiguous} 
    Once $p$ finishes participating in the construction of $\BB_{t}$, 
    it forms a historical view of the first $t$ blackboards 
    $\BB^{(p,t)} = (\BB_1^{(p,t)},\ldots,\BB_t^{(p,t)})$ 
    such that for every $t'\in [t]$, $i\in [n]$, $r\in [m(t)]$,
    \[
    \BB_{t'}^{(p,t)}(r,i) \left\{
    \begin{array}{l@{\hspace{1cm}}l}
    = \BB_{t'}(r,i)       & \mbox{if $\last(i)\neq (t',r)$,}\\
    \in \{\BB_{t'}(r,i), \bot\} & \mbox{otherwise.}
    \end{array}\right.
    \]
    Moreover, $\BB_{t'}^{(p,t)}(r,i) \neq \BB_{t'}(r,i)$ for at most $f$ tuples 
    $(t',r,i)$.
    
    \item If $q$ writes any non-$\bot$ value to $\BB_{t+1}$,
    then by the time any player $p$ fixes $\BB^{(p,t+1)}$, 
    $p$ will be aware of $q$'s view $\BB^{(q,t)}$ of the history up to blackboard $t$.
    
    \item The latency of constructing $\BB^{(p,t)}$ is $O(\sum_{t'\leq t} m(t'))$.
\end{enumerate}
\end{theorem}

\subsection{Coin Flipping and Fraud Detection}\label{subsect:coinflipping-frauddetection}

The King-Saia~\cite{KingS2016,KingS2018} and Huang et al.~\cite{HuangPZ22} protocols 
rely on the fact that the $f$ corrupt players, 
being a small minority, must collectively generate coin flips
whose sum is conspicuously large, as they must often \emph{counteract} the coin flips of $n-2f$ good players.\footnote{Recall that the adversary can fail to deliver messages of up to $f$ players in a timely fashion, so there can be as few as $f + (n-2f)$ players fully participating in any given blackboard/coin-flip.}
At the end of the $t$th iteration of 
$\BrachaAgreement$, the~\cite{KingS2016,HuangPZ22} protocols call $\CoinFlip$, 
which populates the blackboard $\BB_t$ 
with random $\{-1,1\}$ coin flips.  
Define $X_i(t)$ to be the sum 
of the coin flips in
$\BB_t(\cdot,i)$ generated by player $i$.
(Recall that 
the players are partitioned
into $\Akeep,\Aflip$, 
where those in $\Akeep$ will keep the majority value $v^*\in\{-1,1\}$, 
regardless of the outcome of $\CoinFlip$.
Nonetheless, \emph{every} player in 
$\Akeep\cup \Aflip$ 
participates in the $\CoinFlip$ protocol.)

At the very least the adversary wants at 
least one player in $\Aflip$ 
to believe the outcome of the global coin is $\sigma(t) = -v^*$, which is called the \emph{adversarial direction}.
Let $\Sigma_G(t)$ and $\Sigma_B(t)$ be the sum of the good and bad (corrupt)
coin flips written to $\BB_t$.  If $\sgn(\Sigma_G(t))=\sigma(t)$ then the adversary is happy, and if $\sgn(\Sigma_G(t))=-\sigma(t)$ then the adversary needs to counteract the good coin flips and get the total sum $\Sigma_G(t)+\Sigma_B(t)$ in the interval $[-f,f]$ in order for at least one player believe the coin flip outcome (sign of the sum) is $\sigma(t)$.
Thus, 
\[
|\Sigma_B(t)| \geq \max\{0, -\sigma(t)\Sigma_G(t)-f\}.
\]
$\Sigma_G(t)$ is the sum of at least $m(n-2f)=\Omega(mn)$ coin flips; suppose for simplicity that 
$\Sigma_G(t)$ is the sum of \emph{exactly} this many flips.
When $m\gg n > 3f$,  
$f$ is much smaller than the standard deviation of $\Sigma_G(t)$, so let us also ignore the ``$-f$'' term for simplicity.  
By symmetry, $\Sigma_G(t)$ is positive and negative with equal probability, so up to these simplifications,
$\E[\Sigma_B(t)^2] \geq \frac{1}{2}\E[\Sigma_G(t)^2] = \frac{1}{2}m(n-2f)$.\footnote{It is a small abuse of notation to measure the expectation of 
$\Sigma_B(t)^2$ since it has no well defined distribution.  The expectation is naturally w.r.t.~\emph{any} fixed adversarial strategy that convinces at least one player that the outcome of the global coin is $\sigma(t)$.}
On the other hand, if these $f$ bad players 
\emph{were} 
flipping fair coins 
then $\E[\Sigma_B(t)^2] \leq mf$.  

The statistics tracked by Huang et al.~\cite{HuangPZ22} are pairwise \emph{correlations} and individual \emph{deviations} over a series of calls to $\CoinFlip$.
\begin{align*}
    \CORR(i,j) &= \ang{X_i,X_j} = \sum_t X_i(t)X_j(t),\\
    \DEV(i)     &= \ang{X_i,X_i} = \sum_t (X_i(t))^2.
\end{align*}
Note that $(\Sigma_B(t))^2$ can be 
decomposed into terms that contribute
to $\CORR(i,j)$ scores ($X_i(t)X_j(t)$, $i\neq j$)
and $\DEV(i)$ scores ($X_i(t)^2$).
When $f<n/4$ there is a \emph{gap} between $\frac{1}{2}m(n-2f)$ and $mf$, which implies that after a sufficient number of iterations, 
either some bad player $i$ has an unusually large $\DEV(i)$ score, 
or two bad players $i,j$ have an unusually large $\CORR(i,j)$ score~\cite[Lemma 5]{HuangPZ22}.
\emph{Unusually large} here means one beyond what any good players 
flipping fair coins could generate, with high probability.  
If $\CORR(i,j)$ is unusually large, 
it follows that \emph{at least one} of 
$i,j$ must be bad.  The 
Huang et al.~\cite{HuangPZ22} protocol ``blacklists'' players according to a fractional matching in a graph on $[n]$ weighted by correlation and deviation scores, which ensures that bad players are blacklisted at the same or higher rate than the good players.

\subsection{Misfires and Dead Ends}

Our goal is to improve the resilience from $f<n/4$ to the optimal $f<n/3$.  As $f$ tends towards $n/3$, many natural statistics worth tracking lose traction, and ``$n/3$'' is the point at which coin-flipping games become perfectly balanced between the influence of $n-2f$ good and $f$ bad players.
For example, when $f=n/4$, $mf=\frac{1}{2}m(n-2f)$, and bad players may 
not be detected by tracking $\DEV(i)$ and $\CORR(i,j)$ scores alone.
When $n=3f+1$, we can assume $n-f=2f+1$ players fully participate in the coin-flipping protocol, at least $f+1$ of which are good and at most $f$ of which are bad.
To illustrate why this is a uniquely difficult setting to perform fraud detection, consider a simple \textbf{Mirror-Mimic} strategy deployed by the adversary.

\paragraph{Mirror-Mimic Strategy.}  When $\sgn(\Sigma_G(t)) = -\sigma(t)$, the adversary sets $\Sigma_B(t) = -\Sigma_G(t)$ (mirror).
When $\sgn(\Sigma_G(t)) = \sigma(t)$, it 
sets $\Sigma_B(t) = \Sigma_G(t)$ (mimic).
There is some flexibility in the mirror case as it only needs $\Sigma_B(t)+\Sigma_G(t)$ to hit the interval $[-f,f]$.  In any case, we do not expect to see large good-good $\CORR(i,j)$ scores outside of 
random noise,
nor large bad-bad correlations 
since they are mirroring/mimicking the distribution of good players.  Because the mirror/mimic cases occur about equally often, the aggregate positive correlations 
between good and bad players in the mimic case
and negative correlations between good and bad players in the mirror case cancel out.
Thus, against the mirror-mimic adversary, 
tracking pairwise correlations \emph{alone} seems 
insufficient to detect fraud.

\paragraph{$\sigma$-Correlation.} When we attempt to flip a global coin, the good players are generally unaware of the adversarial direction 
$\sigma(t)$\footnote{(if they all knew what it was, there would be no need to flip a coin)} but we can ensure that $\sigma(t)$ eventually becomes known, and can estimate \emph{$\sigma$-correlation} over the long term.
In the context of \BrachaAgreement, $\sigma(t)$ 
should be defined as:
\[
\sigma(t) = \left\{
\begin{array}{ll}
-v^* & \mbox{ if $\Akeep\neq \emptyset$ keeps the majority value $v^*\in\{-1,1\}$,}\\
0   & \mbox{ if $\Akeep = \emptyset$.}
\end{array}
\right.
\]
Define the $\sigma$-correlation score as:
\[
\OCORR(i) = \ang{\sigma, X_i} = \sum_t \sigma(t)X_i(t).
\]
Note that good players flip fair coins, 
so values of $\OCORR(i)$ that are inconsistent with random noise should indicate that $i$ is corrupt.  However, it is rather easy for the adversary to keep $\OCORR(i)$ scores close to zero for corrupt players as well.  Regardless of $\sgn(\Sigma_G(t))$, set $\Sigma_B(t) = -\Sigma_G(t)$, 
or at least put $\Sigma_B(t)+\Sigma_G(t)$ in the interval $[-f,f]$.
One may easily verify that this strategy is consistent with 
mirror-mimic when $\sgn(\Sigma_G(t))=-\sigma(t)$, but that it prescribes exactly the \emph{opposite} behavior when $\sgn(\Sigma_G(t))=\sigma(t)$!
In fact, there is no general strategy for setting $\Sigma_B(t)$ as a function of $\Sigma_G(t)$ and $\sigma(t)$ 
that keeps \emph{all} 
$\CORR(i,j)$ and $\OCORR(i)$ 
scores close to zero.\footnote{When $\sgn(\Sigma_G(t))=-\sigma(t)$, the adversary is \emph{forced} to set $\sigma(t)\Sigma_B(t) \geq -\sigma(t)\Sigma_G(t)$, increasing the aggregate $\sigma$-correlation of bad players and increasing the aggregate negative good-bad correlation.  When $\sgn(\Sigma_G(t))=\sigma(t)$, the adversary can choose to reverse either of these trends and exacerbate the other, i.e., reduce good-bad negative correlations but increase bad $\sigma$-correlations, or reduce bad $\sigma$-correlations but increase good-bad negative correlations.}

Tracking $\CORR(i,j)$ and $\OCORR(i)$ scores seems to be a winning combination, that will eventually let us blacklist individual players for having large $\OCORR(i)$ scores, or pairs of players for having large $-\CORR(i,j)$ scores.  In the latter case, we are blacklisting good and bad players at the same rate, which is fine so long as good players retain their slim majority ($f+1$ vs. $f$ initially) among any set of $n-f=2f+1$ participating players.

\paragraph{A Scheduling Attack.}  
There is a serious flaw in the reasoning above.
Recall that $\sigma(t)\in \{-1,0,1\}$,
where $\sigma(t)=0$ means that in \BrachaAgreement, the population $\Akeep=\emptyset$ committed to keeping the true majority value $v^*$ is empty.  
The value $v^*$ is determined by the scheduling of messages in Line~\ref{line:Bracha1}.
Whether $v^*$ becomes \emph{known} to any particular player in Line~\ref{line:Bracha2} is generally at the discretion of the adversarial scheduler. 
Thus, in general the adversary can control whether $\Akeep=\emptyset$, and hence 
whether $\sigma(t)=-v^*$ or $0$.  Moreover, because the protocol is asynchronous, it can even do so \emph{after} $\BB_t$ is populated with coin flips.\footnote{One player $p$ will be allowed to set $v_p=v^*$ in Line~\ref{line:Bracha2}, and $p$ will execute Line~\ref{line:Bracha3} slowly,
either validating its own message among the first $n-f$ or not, at the discretion of the scheduler.  In this way the scheduler decides if $\Akeep=\{p\}$ or $\emptyset$.}

These observations 
give rise to the following attack.  The adversary targets two good players $i_0,i_1$.  When 
$\CoinFlip$ 
is initiated the adversary has two choices for $\sigma(t) \in \{-v^*,0\}$ and can decide which late in the game.  If $\sgn(X_{i_0}(t))=\sgn(X_{i_1}(t)) = -v^*$, it sets $\sigma(t)=-v^*$; otherwise it sets $\sigma(t)=0$.  In general it makes sure $\Sigma_B(t)+\Sigma_G(t) \in [-f,f]$ so it can force roughly equal numbers of 
players to have \CoinFlip{} return $-1$ and $1$.  
Players $i_0,i_1$ will show unusually large $\sigma$-correlation and be blacklisted, and any other blacklisting (from negative correlations) will apply equally to good and bad players. 
At this point the \emph{corrupt} players have now attained a slim majority,
and are entirely content to let further blacklisting hurt good and bad players equally.

The problem here is that 
$\sigma(t)$ and $X_{i_0}(t)$ are 
\emph{not} independent.
In reality $\sigma(t)$ can be chosen maliciously \emph{after} $X_{i_0}(t)$ is known.

\paragraph{A Finger on the Scale.}
The issue with the previous scheme is that the notions of $\sigma(t)$, the population $\Akeep$, 
and even $v^*$ are too indeterminate.  
On the other hand, 
if any good player $p$ finds itself with $x_p\in[1,f]$, 
there is nothing indeterminate from $p$'s perspective
about the fact that 
$\{p\}\subseteq \Akeep\neq \emptyset$
or that $\sigma(t) = -v^*$.
This leads to a natural question: why should $p$
participate in the $\CoinFlip$ protocol 
\emph{as if it were ignorant of the 
desired outcome $v^*$}?  
Why not ``put a finger on the scale'' and 
just write $v^*$ to every entry 
in column $\BB_t(\cdot,p)$?
(We would naturally refrain from judging such special columns according to statistical tests, e.g. deviations and correlations.)

The problem with this simple minded scheme
is that if $|\Akeep|$ is small,
the adversary has the discretion to suppress
or allow $p\in \Akeep$ to write its column, or any prefix thereof.  
This allows for a mirror-mimic type attack, 
in which the sum of $\BB_t$ always 
lies in $[-f,f]$, 
and yet there are no negative correlations 
in aggregate between 
good players flipping fair coins and bad players.

\subsection{Overview of the Protocol}

To simplify the description of the 
coin-flipping problem, 
in \autoref{subsect:coinflipping-frauddetection}
we originally stated the 
adversary chooses $\sigma(t)=-v^*$, 
and has the goal 
to convince \emph{one} player to believe the
output of \CoinFlip{} is $\sigma(t)$. 
This is too conservative, 
and in fact makes the coin-flipping 
problem needlessly difficult.

\medskip
Consider the sizes of the sets $\Akeep$ and $\Aflip$.
There are two relevant cases to consider:
\begin{description}
\item[Case {$|\Akeep| \in [0,f]$}.]  
When $\Akeep\neq\emptyset$,
$\sigma(t)=-v^*$ is defined, and it would be bad for the adversary if everyone 
agreed the output of \CoinFlip{} were $v^*$, i.e., $\sgn(\Sigma^{(p)}) = v^*$ for all $p\in[n]$.
In fact, it would be equally bad for the adversary if $\sgn(\Sigma^{(p)}) = -v^*$ for all $p$.  
Since $\Aflip\geq n-f$, 
the next iteration of \BrachaAgreement{} 
would end in agreement since now a supermajority of 
$n-f \geq (n+f+1)/2$ holds the value $-v^*$.
To summarize,
when $|\Akeep|\leq f$, 
it is \emph{critical} for the adversary 
to create \emph{disagreement} on the outcome 
of the global coin flip.
\item[Case $|\Akeep| \geq f+1$.] If this is the case, then some kind of ``finger on the scale'' strategy should force the outcome of the coin flip to be $v^*$.  Any player that validates
the state of $n-f$ players must necessarily validate the state of some $p\in \Akeep$,
and hence learn the value of $v^*$.  If any player that knows $v^*$ writes only $v^*$ to its entries in the blackboard, this will surely be the outcome of the global coin flip.
\end{description}

In light of this dichotomy on the size of $\Akeep$, 
we design a coin-flipping protocol that 
(i) forces all players to see the same outcome $v^*$ whenever 
$|\Akeep|\geq f+1$ 
--- thereby letting \BrachaAgreement{} terminate --- 
or 
(ii) reverts to a more standard 
collective coin-flipping game in which the adversary is obligated to land the sum 
in the interval $[-f,f]$.
Because of the certainty of the outcome in case (i) 
and the specific strategy forced upon the adversary in case (ii),
fraud can now be detected by tracking just one statistic: 
the \emph{weighted} correlation 
scores $\CORR(i,j)$ between all pairs of players.

\section{A Protocol with Optimal Resilience}\label{sect:protocol}

The protocol consists of $O(f)$ \emph{epochs}. 
Each epoch consists of $T$ 
iterations of \BrachaAgreement, 
the last step of which is to execute
\CoinFlip{} (Algorithm~\ref{alg:coinflip}).
The $t$th execution of
\CoinFlip{} constructs two blackboards
$\BB_{2t-1}$ and $\BB_{2t}$, 
where the odd blackboards have $\tilde{O}(\sqrt{m})$ rows and
the even blackboards have $m$ rows,
for an $m\gg n$ to be determined.
We write $n=(3+\epsilon)f$, where we assume without loss of generality that $\epsilon\in [1/f,1/2]$.

Each player $i$ has a \emph{weight} 
$w_i \in [0,1]$, initially $1$ and non-increasing over time.  
(The coin flips generated by player 
$i$ will be weighted by $w_i$.)
At the end of each epoch the weights 
of some players will be reduced, 
which one can think of 
as \emph{fractional blacklisting}.  
Let $(w_{i,k})$ denote the weights used
throughout epoch $k$, and $(w_i)$ be the current weights if $k$ is understood from context.  We guarantee that Invariant~\ref{inv:weights} is maintained, with high probability.
Let $G$ and $B$ be the good and bad (Byzantine)
players at a given point in time.

\begin{invariant}\label{inv:weights}
At all times,
\[
\sum_{i\in G} (1-w_i) \leq \sum_{i\in B} (1-w_i) + \epsilon^4 f.
\]
\end{invariant}

\paragraph{Organization.}
In \autoref{sect:implementation-of-coinflip} we explain how \CoinFlip{} is implemented using the \IteratedBlackboard{} protocol~\cite{HuangPZ22}.
In \autoref{sect:neg-correlations} we prove that if the adversary persistently manipulates the outcome of the \CoinFlip{} protocol, 
then there will be a detectable negative correlation between 
some bad player and some good player.
In \autoref{sect:weight-reduction} we give the procedure for reducing weights between epochs, and prove that it maintains Invariant~\ref{inv:weights} with high probability.
(The weight reduction routine uses the same fractional matching method as~\cite{HuangPZ22}.) 
In \autoref{sect:bounding-error} 
we prove that agreement is reached 
after $O(f)$ epochs, with high probability.

\subsection{Implementation of \CoinFlip}\label{sect:implementation-of-coinflip} 
\newcommand{\writeval}{\operatorname{val}}

Consider 
the $t$th iteration of \BrachaAgreement{} and
the $t$th execution of \CoinFlip.
When each player $p$ begins executing \CoinFlip, 
it has a value $v_p \in \{-1,1,\bot\}$, where
$v_p \in \{-1,1\}$ indicates that $v_p = v^*$ 
is \underline{\emph{the}} majority value at Line~\ref{line:Bracha2} of \BrachaAgreement,
and $v_p =\bot$ indicates that $p$ did not learn
the majority value and will adopt the output of \CoinFlip{}
as its value going into the next iteration of \BrachaAgreement.

Recall that $x_p$ is the number of $v^*\in\{-1,1\}$ messages validated by $p$ in Line~\ref{line:Bracha3} of \BrachaAgreement{} and $\Akeep$ is the set of all $p$
such that $x_p \in [1,f]$ before executing \CoinFlip.
The first stage of \CoinFlip{} is to populate a 
blackboard $\BB_{2t-1}$ that will help end the game 
quickly if $|\Akeep|\geq f+1$ and 
cause no harm if $|\Akeep| \in [0,f]$.
The contents of $\BB_{2t-1}$ are used to generate a \emph{bias}.
The second stage of \CoinFlip{} populates a blackboard
$\BB_{2t}$ with random values in $\{-1,1\}$.
Let 
\[
\Sigma = \sum_{j,q} w_q \cdot \BB_{2t}(j,q)
\]
be the \emph{weighted} sum of the contents of $\BB_{2t}$ (mapping $\bot$ to 0).  
The output of \CoinFlip{} is 
$\sgn(\bias + \Sigma)$.
Due to the scheduling power of the adversary, each player
has a slightly different view of these two blackboards.
Naturally $p$ outputs $\sgn(\bias^{(p)} + \Sigma^{(p)})$, 
where a superscript of $(p)$ in any variable 
indicates $p$'s opinion of its value.

\begin{algorithm}[H]
    \caption{\CoinFlip() \ \emph{from the perspective of player $p$}\label{alg:coinflip}}
\begin{algorithmic}[1]
\Require{$v_p\in \{v^*,\bot\}$, $v^*\in \{-1,1\}$, and 
$t\geq 1$ is current iteration of \BrachaAgreement.}
\State \emph{\textbf{Stage 1:}\istrut{4}}
\State \ReliableBroadcast{} $v_p$ and wait for $n-f$ messages to be validated from some set $S_p$ of players.
\State $\writeval_p \gets \left\{
\begin{array}{ll}
v^* & \mbox{ if $v_q=v^*$ for some $q\in S_p$,}\\
0 & \mbox{ if $v_q=\bot$ for all $q\in S_p$.}
\end{array}\right.$
\State Construct $\BB_{2t-1}$, writing $\writeval_p$ to every cell in column $p$.
\State \emph{\textbf{Stage 2:}\istrut{4}}
\State Construct $\BB_{2t}$, writing independent coin flips in $\{-1,1\}$ to cells in column $p$.
\State $\bias^{(p)} \gets \sum_{j,q} \BB_{2t-1}^{(p,2t)}(j,q)$ 
\Comment{Substitute $\bot = 0$}
\State $\Sigma^{(p)} \gets \sum_{j,q} w_q\cdot \BB_{2t}^{(p,2t)}(j,q)$
\Comment{Substitute $\bot = 0$}
\State \Return($\sgn(\bias^{(p)} + \Sigma^{(p)})$)
\end{algorithmic}
\end{algorithm}

The even and odd blackboards 
$\BB_{2t}$ and $\BB_{2t-1}$
have the following number of rows
\begin{align*}
m(2t)   &= m,\\
m(2t-1) &= m_0 = \sqrt{m\cdot c\ln n}.
\end{align*}
Here $c$ is a parameter, 
and \emph{with high probability} 
means an event holds
with probability $1-n^{-\Omega(c)}$.

\begin{lemma}\label{lem:bias-vs-Akeep}
If $|\Akeep|\geq f+1$
then
$v^*\cdot \bias 
> (n-f)m_0 
> \frac{2}{3}n\sqrt{m\cdot c\ln n}$.
If $|\Akeep| = 0$ then 
$\bias = 0$.
If $|\Akeep|\in [1,f]$ then 
$v^*\cdot \bias \in [0,nm_0]$, 
and can be selected by the adversary.
\end{lemma}

\begin{proof}
If $|\Akeep|\geq f+1$ then
for every player $p$, there exists a $q\in S_p$
with $q\in \Akeep$ and $v_q=v^*$,
hence $\writeval_p=v^*$.
By \autoref{thm:blackboard}, the number
of values written to 
$\BB_{2t-1}$ is at least
$(n-f)m_0$.
and hence 
$v^*\cdot \bias \geq (n-f)m_0 > \frac{2}{3}n\sqrt{m\cdot c\ln n}$.

If $|\Akeep| = 0$ then 
every player will set $\writeval_p = 0$
hence $\bias=0$.
\end{proof}

In the second stage of \CoinFlip, 
the players will populate $\BB_{2t}$ with coin flips
in $\{-1,1\}$.  
Redefine $X_i(t)$ be the sum of 
all non-$\bot$ entries in column $\BB_{2t}(\cdot, i)$.
By Chernoff bounds, if $i$ is uncorrupted then $|X_i(t)| \leq \XMAX$ 
with high probability, where
\[
\XMAX = m_0 = \sqrt{m\cdot c\ln n}.
\]
We will force $|X_i(t)| \leq \XMAX$ 
to hold with probability 1 by 
rounding $X_i(t)$ to $\pm \XMAX$ if it lies outside $[-\XMAX,\XMAX]$.

\begin{lemma}\label{lem:Akeep-large}
If $|\Akeep|\geq f+1$, the output of $\CoinFlip$ will be $v^*$ for all players, with high probability.
\end{lemma}

\begin{proof}
By \autoref{lem:bias-vs-Akeep}, 
$|\bias|\geq (n-f)m_0 
> \frac{2}{3}n\sqrt{m\cdot c\ln n}$.
The number of good coin flips in $\BB_{2t}$ 
is between $m(n-2f)$ and $mn$, which we model
as a martingale with an optional stopping time
controlled by the adversary. 
By Azuma's inequality, the sum
of all $\Theta(mn)$ good coin flips
is $\sqrt{mn\cdot c\ln n}$ in 
absolute value, with high probability.
Due to the $\XMAX$ ceiling, the contribution
of corrupt players to the sum 
is at most 
$f\XMAX < \frac{1}{3}n\sqrt{m\cdot c\ln n}$ in absolute value.
Since 
$\frac{1}{3}n\sqrt{mc\ln n} + \tilde{O}(\sqrt{mn}) + f 
< \frac{2}{3}n\sqrt{mc\ln n}$, 
the contribution of corrupt and non-corrupt players 
will be much smaller than $\bias$, with high probability,
and $\sgn(\bias^{(p)} + \Sigma^{(p)}) = v^*$ for all $p$.
\end{proof}

\begin{lemma}\label{lem:Akeep-small}
If $|\Akeep| \leq f$,
and for some player $p$, 
$\bias^{(p)} + \Sigma^{(p)} \not\in [-f,f]$,
all players will \textbf{decide} 
the value $\sgn(\bias^{(p)} + \Sigma^{(p)})$
by the next iteration of
\BrachaAgreement.
\end{lemma}

\begin{proof}
By \autoref{thm:blackboard}, two players
$p$ and $q$ disagree in at most $f$ locations
in $\BB_{2t-1}$ and $\BB_{2t}$.
Since the absolute value of any cell 
in either matrix is at most 1, 
if $\bias^{(p)} + \Sigma^{(p)} \not\in [-f,f]$
then for any $p,q$, $\sgn(\bias^{(p)} + \Sigma^{(p)})
= 
\sgn(\bias^{(q)} + \Sigma^{(q)})$.
Thus, at the end of this iteration of \BrachaAgreement,
at least $|\Aflip|\geq n-f \ge (n+f+1)/2$ 
players will hold a supermajority and reach agreement
in the next iteration of \BrachaAgreement.
\end{proof}

\subsection{Negative Correlations}\label{sect:neg-correlations}

Hereafter $t\in [T]$ refers to the iteration index 
within the current epoch $k$, 
and $(w_i)=(w_{i,k})$ 
is the current weight vector used in epoch $k$.  
In each epoch we track 
weighted pairwise correlations, defined as:
\[
\CORR(i,j)  
    = \ang{w_i X_i, w_j X_j}
    = w_iw_j \sum_{t\in [T]} X_i(t) X_j(t).
\]

Recall that $n=(3+\epsilon)f$, 
$\epsilon \in [1/f,1/2]$,
and $G,B$ are the good and bad players.
\begin{lemma}\label{lem:n/3-gap-lemma}
Suppose the weight vector $(w_i)_{i\in [n]}$ used in an epoch
satisfies \autoref{inv:weights}, but \BrachaAgreement{} fails to 
reach agreement within the epoch.  
Let $m=\Omega(n\ln n/\epsilon^4)$ and $T=\Omega(n^2\ln^3 n/\epsilon^4)$.
Then, with high probability,
\begin{enumerate}[itemsep=0pt]

    \item Every pair of distinct $i,j \in G$ has $-\CORR(i,j) \leq w_iw_j\beta$, where $\beta = m\sqrt{T(c\ln n)^3}$.
    \item If the adversary does not corrupt any new players 
    during the epoch, then 
    \[
        \sum_{(i,j) \in G \times B} \max \{ 0, -\CORR(i,j) - w_iw_j\beta \} \geq \frac1{8}\epsilon^2fmT.
    \]
\end{enumerate}
\end{lemma}

\begin{proof} [Proof of Part 1.]
Fix an iteration $t \in [T]$.
If $i\in G$, let $\delta_{t,r}^i \in\{-1,0,1\}$ be 
the outcome of $i$'s $r$th coin flip in iteration $t$, 
where 0 indicates the coin was never flipped.
For any $r, s \geq 1$ and $i,j\in G$,
$\E[\delta_{t,r}^i\delta_{t,s}^j] = 0$ since each 
of $\delta_{t,r}^i,\delta_{t,s}^j$ is either 0 
or a fair coin flip independent of the other.
By linearity of expectation this implies $\E[X_i(t)X_j(t)]=0$ as well.

Consider the martingale 
$(S_t)_{t\in[0,T]}$, 
where $S_0 = 0$ and $S_t = S_{t-1} + X_i(t)X_j(t)$.
For any $t$, $|S_t - S_{t-1}| \leq \XMAX^2$.
By Azuma's inequality, 
$|S_T| \leq \XMAX^2\sqrt{T\cdot c\ln n}
= m\sqrt{T(c\ln n)^3}$ with high probability.
Therefore, by a union bound, 
for every pair of distinct 
$i,j\in G$,
$-\CORR(i,j) \leq w_i w_j m\sqrt{T(c\ln n)^3} = w_iw_j\beta$ 
with high probability.
\end{proof}

Part 2 of 
\autoref{lem:n/3-gap-lemma} 
will be proved following 
Lemmas~\ref{lem:bias-properties}--\ref{lem:biasSG}.
It only applies to epochs in which
the adversary corrupts no one, so we shall assume that $G,B$ are stable throughout the epoch.

In the \CoinFlip{} algorithm, 
the construction of $\BB_{2t-1}$ \emph{logically} precedes the 
construction of $\BB_{2t}$, 
but because of asynchrony 
some of the contents of $\BB_{2t-1}$ 
may actually depend on the 
coin flips written to $\BB_{2t}$.\footnote{The construction of $\BB_{2t}$ can proceed as soon as $n-f$ players are finished
with $\BB_{2t-1}$.  Thus, a group of $f$ slow and corrupt players can 
choose whether to perform their last write in $\BB_{2t-1}$ based on the 
contents of $\BB_{2t}$.}
We eliminate these mild 
dependencies as follows.  
Suppose that $\hat{p}$ is the \emph{first} player in $G$ 
to fix its historical view $\BB^{(\hat{p},2t-1)}$.
At this moment, 
define $\biasbar(t)$ as
\begin{align*}
\biasbar(t) &= \sum_{j,q} \BB_{2t-1}^{(\hat{p},2t-1)}(j,q)
& \mbox{(Treating $\bot$ as 0)}
\end{align*}
Write $\Sigma(t) = \Sigma_G(t) + \Sigma_B(t)$, 
where $\Sigma_G(t)$ and $\Sigma_B(t)$ are the sum of 
coin flips of the good and bad players, respectively. 

\begin{lemma}
\label{lem:bias-properties}
In any iteration $t$,
\begin{enumerate}
    \item For any player $q$, $|\biasbar(t) - \bias^{(q)}(t)| \leq f$.
    \item $\E[\biasbar(t)\Sigma_G(t)]=0$.
    \item If \BrachaAgreement{} does not terminate by iteration $t+1$, then 
\[
    -\Sigma_G(t)\Sigma_B(t) 
    \geq \Sigma_G(t)^2 + \biasbar(t)\Sigma_G(t) - 2f|\Sigma_G(t)|.
\]
\end{enumerate}
\end{lemma}
\begin{proof}
\emph{Part 1.}
By \autoref{thm:blackboard}, 
$\BB_{2t-1}^{(\hat{p},2t-1)}$
and 
$\BB_{2t-1}^{(q,t')}$ disagree in at most $f$ cells, 
for any $q\in [n]$ 
and $t'\geq 2t-1$, 
hence $|\biasbar(t) - \bias^{(q)}(t)| \leq f$.
\emph{Part 2.}
By definition, $\biasbar(t)$ is fixed before any good players have written anything to $\BB_{2t}$.  Thus $\E[\biasbar(t)\Sigma_G(t)]=\biasbar(t)\cdot \E[\Sigma_G(t)]=0$.
\emph{Part 3.}
By \autoref{lem:Akeep-small}, 
if the adversary avoids termination then
$\sgn(\bias^{(p)}(t) + \Sigma^{(p)}(t)) \neq \sgn(\bias^{(q)}(t) + \Sigma^{(q)}(t))$ for two players $p,q$.
Since $|\biasbar(t) - \bias^{(p)}(t)|\leq f$
and $|\Sigma(t)-\Sigma^{(p)}(t)|\leq f$,
it follows from $\Sigma(t)=\Sigma_G(t) + \Sigma_B(t)$ that
\begin{align}
-2f &\leq \biasbar(t) + \Sigma_G(t) + \Sigma_B(t) \leq 2f.\nonumber
\intertext{Rearranging terms, we have both}
-\Sigma_B(t) &\geq -2f + \biasbar(t) + \Sigma_G(t)\nonumber
\intertext{and}
\Sigma_B(t) &\geq -2f - \biasbar(t) - \Sigma_G(t).\nonumber
\intertext{Depending on $\sgn(\Sigma_G(t))$,
we multiply the first inequality 
by $\Sigma_G(t)\geq 0$ or the second by $-\Sigma_G(t)\geq 0$,
which implies the following.}
-\Sigma_G(t)\Sigma_B(t) &\geq \Sigma_G(t)^2 + \biasbar(t)\Sigma_G(t) - 2f|\Sigma_G(t)|.\label{eqn:GB-product}
\end{align}
\end{proof}

Since $\Var(|\Sigma_G(t)|) = \E[\Sigma_G(t)^2]-\E[|\Sigma_G(t)|]^2\geq 0$, 
$\E[|\Sigma_G(t)|]\leq \sqrt{\E[\Sigma_G(t)^2]}$.
Lemmas~\ref{lem:weight-bound}--\ref{lem:biasSG} analyze
the terms of (\ref{eqn:GB-product}).

\begin{lemma}
\label{lem:weight-bound}
For any $\hat{G} \subseteq G$ with $|\hat{G}| = n-2f = (1+\epsilon)f$, 
$\sum_{i \in \hat{G}} w_i^2 \ge \frac12\epsilon^2f$.
\end{lemma}
\begin{proof}
We compute:
\begin{align*}
    \sum_{i \in \hat{G}} w_i &= |\hat{G}| - \sum_{i \in \hat{G}} (1-w_i) 
    \geq |\hat{G}| - \sum_{i \in B} (1-w_i) - \epsilon^4 f \geq |\hat{G}| - \left(1+\epsilon^4\right) f 
    = \left(\epsilon-\epsilon^4\right)f 
    \geq \frac{7\epsilon}{8}f.
\end{align*}
The first inequality follows from \autoref{inv:weights} and the fact that the total weight deduction of $\hat{G}$ is at most that of $G$. The second inequality follows from $w_i \in [0,1]$, so the total weight deduction of $B$ is at most $f$. The equality follows from $|\hat{G}|=n-2f = (1+\epsilon)f$. Finally, the last inequality follows from the assumption that $\epsilon \leq 1/2$. 
Consequently: 
\begin{align*}
    \sum_{i \in \hat{G}} w_i^2 
    = |\hat{G}| \sum_{i \in \hat{G}} w_i^2\frac{1}{|\hat{G}|} \geq |\hat{G}|\left(\sum_{i \in \hat{G}} w_i\frac{1}{|\hat{G}|} \right)^2 
    \geq |\hat{G}|\left(\frac{7\epsilon}{8(1+\epsilon)}\right)^2  
    \geq \frac{1}{2}\epsilon^2 f \, ,
\end{align*}
where the first inequality is Jensen's inequality, 
the middle inequality is from above, 
and the last inequality follows from 
$|\hat{G}| = (1+\epsilon)f$
and the 
assumption $\epsilon \leq \frac{1}{2}$.
\end{proof}

\begin{lemma}[Cf.~\cite{HuangPZ22}]
\label{lem:expected-sum-squared}
No matter how the coin flips of $G$ 
are scheduled in iteration $t$, 
$\E[\Sigma_G(t)^2] \geq \frac{1}{2}\epsilon^2 mf$.
\end{lemma}
\begin{proof}
The good players write between $m(n-2f)$ and $mn$ coin flips to $\BB_{2t}$, 
at the adversary's discretion. 
For $r\in [0,2mf]$, 
let $S_r$ be the sum of the first $m(n-2f)+r$ coin flips 
generated by the good players. 
Then $\E[\Sigma_G(t)^2] = \E[S_{2fm}^2]$, 
which we claim is at least $\E[S_0]$.
In general 
$S_r = S_{r-1} + w_i\delta_r$, 
where
$\delta_r \in \{-1,1\}$ if the adversary
lets player $i$ flip the next coin and
$\delta_r = 0$ if the adversary 
chooses to stop allowing coin flips.
If $\delta_r=0$ then $S_r=S_{r-1}$ and 
if $\delta_r\in\{-1,1\}$ then
\[
    S_r^2 = \begin{cases}
    (S_{r-1} + w_i)^2 = S_{r-1}^2 + 2w_iS_{r-1} + w_i^2 &\textrm{with probability } \frac{1}{2}, \\
    (S_{r-1} - w_i)^2 = S_{r-1}^2 - 2w_iS_{r-1} + w_i^2 &\textrm{with probability } \frac{1}{2}.
    \end{cases}
\]
Thus, $\E[S_r^2 \mid \delta_r \neq 0] 
= S_{r-1}^2 + w_i^2 > S_{r-1}^2$,
and in general,
$\E[S_{r}^2] \geq \E[S_{r-1}^2] \geq \E[S_0^2]$.
Thus, the adversarial strategy minimizing $\Sigma_G(t)^2$
is to allow as few coin flips as possible, 
and from those $n-2f$ 
players $\hat{G}$ with the smallest weights.
By \autoref{lem:weight-bound} we have
\[
\E[S_0^2] \geq m\sum_{i \in \hat{G}} w_i^2 \geq \frac{1}{2}\epsilon^2 mf.
\]
\end{proof}

\begin{lemma}[Cf.~\cite{HuangPZ22}]
\label{lem:sumSGsquared}
With high probability, 
$\sum_{t\in [T]} \Sigma_G(t)^2 \geq \frac{1}{2}\epsilon^2 mfT 
- mn\sqrt{T(c\ln n)^3}$.
\end{lemma}
\begin{proof}
Consider the sequence $(A_t)_{t \in [T]}$, where $A_0 = 0$, $A_t = A_{t-1} + \Sigma_G(t)^2 - \frac{1}{2}\epsilon^2 mf$. 
By \autoref{lem:expected-sum-squared}, 
$\E[\Sigma_G(t)^2] \geq \frac{1}{2}\epsilon^2 mf$,
so $(A_t)_t$ is a submartingale. 
Since $\Sigma_G(t)$ is a sum of at most $mn$ coin flips, 
$|A_t - A_{t-1}| = \Sigma_G(t)^2 \leq mn\cdot c\ln n$ 
with high probability. 
By Azuma's inequality, with high probability, 
$A_T \geq -(mn c\ln n)\sqrt{Tc\ln n}$
and 
\[
\sum_{t \in [T]} \Sigma_G(t)^2 
= 
\frac{1}{2}\epsilon^2mfT + A_T \geq \frac{1}{2}\epsilon^2mfT - mn\sqrt{T(c\ln n)^3}.
\]
\end{proof}

\begin{lemma}\label{lem:biasSG}
With high probability, 
$\sum_{t \in [T]} \overline{\bias(t)}\Sigma_G(t) \leq mn\sqrt{nT(c\ln n)^3}$.
\end{lemma}
\begin{proof}
By \autoref{lem:bias-properties}, 
$\E[\biasbar(t)\Sigma_G(t)] = 0$
and hence the sequence $(A_t)_{t\in [T]}$ is a martingale,
where $A_0 = 0$ and $A_t = A_{t-1} + \biasbar(t)\Sigma_G(t)$. 
With high probability $|\Sigma_G(t)|\leq \sqrt{mn(c\ln n)}$
and $\biasbar(t) \leq nm_0 = n\sqrt{m(c\ln n)}$, 
hence by Azuma's inequality, 
$\sum_{t \in [T]} \biasbar(t)\Sigma_G(t) \leq 
mn\sqrt{nT(c\ln n)^3}$  
with high probability. 
\end{proof}

We are now equipped to prove the second part of \autoref{lem:n/3-gap-lemma}.

\begin{proof}[Proof of Part 2 of \Cref{lem:n/3-gap-lemma}] 
Recall $G,B$ are the sets of good and bad players, which, by assumption, do not change during the epoch.
\begin{align*}
    - \sum_{(i, j)\in G\times B} \CORR(i, j) &= -\sum_{t\in [T]} \Sigma_G(t)\Sigma_B(t) \\
    &\geq \sum_{t\in [T]} \left(\Sigma_G(t)^2 - 2f|\Sigma_G(t)| 
    + \biasbar(t)\Sigma_G(t) \right) 
    \tag{\autoref{lem:bias-properties}}\\
    &\geq \frac{1}{2}\epsilon^2 mfT-\tilde{O}(mn\sqrt{T}) - 2f\tilde{O}(\sqrt{mn}T) - \tilde{O}(mn\sqrt{nT}) 
    \tag{W.h.p., by Lemmas~\ref{lem:sumSGsquared}, \ref{lem:biasSG}}\\
    &= \left(\frac{1}{2}\epsilon^2 - \tilde{O}\left(\sqrt{\frac{n}{m}}\right)\right)mfT - \tilde{O}(mn\sqrt{nT}) \\
    &\ge \left(\frac{1}{2}\epsilon^2 
    - o(\epsilon^2)\right) mfT - \tilde{O}(mn\sqrt{nT}) 
    \tag{whenever $m = \Omega(n\ln n/\epsilon^4)$} 
    \\
    &\ge\frac14\epsilon^2 fmT 
    \tag{whenever $T=\Omega(n\ln^3 n/\epsilon^4)$} 
\end{align*}

Finally, since $\max\{0,-\CORR(i,j) - w_iw_j\beta\} \geq -\CORR(i,j) - w_iw_j \beta$, 
\autoref{lem:n/3-gap-lemma}(2) 
follows from 
the above inequality 
and the fact that 
\begin{align*}
    \sum_{(i, j)\in G\times B} w_iw_j\beta \leq |G|\cdot |B|\cdot \beta \leq nf\cdot \tilde{O}(m\sqrt{T}) 
    \le \frac{1}{8}\epsilon^2 fmT 
\end{align*}
holds whenever $T = \Omega(n^2\ln^3 n/\epsilon^4)$.
\end{proof}

\subsection{Blacklisting via Fractional Matching}\label{sect:weight-reduction}

When the $T$ iterations of epoch $k$ are complete, 
we reduce the weight vector $(w_i)$ in preparation for epoch $k+1$.
According to \autoref{lem:n/3-gap-lemma}, if a correlation score $-\CORR(i,j)$ is too large, 
$B\cap\{i,j\}\neq\emptyset$ w.h.p., so reducing \emph{both} $i$ and $j$'s weights
by the \emph{same} amount preserves \autoref{inv:weights}.  
With this end in mind, \WeightUpdate{} 
(\autoref{alg:weights-update-real-bb}) 
constructs a complete, vertex- and edge-capacitated graph $H$ on $[n]$, 
finds a fractional maximal matching $\mu$ in $H$, 
then docks the weights of $i$ and $j$ by $\mu(i,j)$, 
for each edge $(i,j)$.  This is essentially
the same as the blacklisting routine of~\cite{HuangPZ22}, except we are paying
attention to large \emph{negative} correlations instead
of large positive correlations and individual deviations.

\begin{definition}[Fractional Maximal Matching]
Let $H=(V,E,c_V,c_E)$ be a graph where $c_V : V\rightarrow \mathbb{R}_{\geq 0}$ are vertex capacities
and $c_E : E\rightarrow \mathbb{R}_{\geq 0}$ are edge capacities.  
A function $\mu : E\rightarrow \mathbb{R}_{\geq 0}$
is a \emph{feasible fractional matching} 
if $\mu(i,j)\leq c_E(i,j)$ and $\sum_{j} \mu(i,j) \leq c_V(i)$.
It is \emph{maximal} if it is not strictly dominated by any feasible $\mu'$. 
\end{definition}

\paragraph{Rounding Weights Down.}
At the end of epoch $k$, player $p$ generates a 
local weight vector
$(w_{i,k+1}^{(p)})_{i\in [n]}$,
which is a function of $(w_{i,k})_{i\in [n]}$ and 
its historical view $\BB^{(p,2kT)}$.  
(There are $2T$ blackboards in each epoch.)
The consensus weight vector $(w_{i,k+1})_{i\in [n]}$
is obtained by everyone 
adopting the weight of $i$ 
according to player $i$'s local view, 
and rounding down if it is too close to zero.

\[
w_{i,k+1} = \left\{\begin{array}{ll}
w_{i,k+1}^{(i)}  & \mbox{ if $w_{i,k+1}^{(i)} > w_{\min} = \frac{\sqrt{n}}{T}$,}\\
0           & \mbox{ otherwise.}
\end{array}\right.
\]
Recall from \autoref{thm:blackboard} 
that if $i$ writes anything to any blackboard in epoch $k+1$,
that every player can deduce what its view $\BB^{(i,2kT)}$
looked like at the end of epoch $k$, and hence what
$w_{i,k+1}^{(i)}$ and $w_{i,k+1}$ are.  By ensuring that
all participating players use exactly the \emph{same} weight function, 
we eliminate one source of potential numerical disagreement.

We will see that the
maximum pointwise disagreement in the local weight vectors $|w_{i,k+1}^{(p)}-w_{i,k+1}^{(q)}|$ is at most $w_{\min}$. 
As a consequence, if any $p$ thinks that $w_{i,k+1}^{(p)}=0$ 
then all players will agree that $w_{i,k+1}=0$.

\newcommand{\ErrT}{\ensuremath{err_T}}

\paragraph{Excess Graph.} The \emph{excess correlation graph} 
$H=(V, E, c_V,c_E)$ used in \autoref{alg:weights-update-real-bb} 
is a complete undirected graph on $V=[n]$, 
capacitated as follows:
\begin{align*}
    c_V(i) &= w_i, \\
    c_E(i, j) &= \frac{8}{\epsilon^2 fmT} \cdot \max\{0, \CORR(i, j)-w_iw_j\beta \},
\end{align*}
where $\beta$ is the quantity from \autoref{lem:n/3-gap-lemma}.
By Part 1 of \autoref{lem:n/3-gap-lemma}, $c_E(i,j)=0$ whenever both $i$ and $j$ are good.

\medskip

\begin{algorithm}[h]
\caption{\WeightUpdate{} \ \ \emph{from the perspective of player $p$}.}\label{alg:weights-update-real-bb} 
\textbf{Output:} Weights $(w_{i,k})_{i\in [n], k\geq 0}$ where $w_{i,k}$ refers to the
weight $w_i$ \emph{after} processing epoch $k-1$, 
and is used throughout epoch $k$.\\
\begin{algorithmic}[1]
\State Set $w_{i, 1}\gets 1$ for all $i$.\Comment{All weights are 1 in epoch 1.}
\For{epoch $k=1, 2, \ldots, \KMAX$}\Comment{$\KMAX = $ last epoch}
\State Execute $T$ iterations of \BrachaAgreement, using weights $(w_{i,k})$ in \CoinFlip. 
Let $\CORR^{(p)}$ be the resulting 
correlation scores known to $p$.
Construct the excess 
correlation graph $H_k^{(p)}$ with capacities:
\begin{align*}
c_{V}(i) &= w_{i, k}, \\[4pt]
c_{E}^{(p)}(i, j) &=  \displaystyle\frac{8}{\epsilon^2 fmT}\cdot \max\left\{ 0, \CORR^{(p)}(i, j) - w_{i, k}w_{j, k}\beta\right\}.
\end{align*}
\State $\mu_k^{(p)} \gets \RisingTide(H_k^{(p)})$ \Comment{A maximal fractional matching}
\State For each $i$, once $w_{i,k}$ is known, set
\[
w_{i, k+1}^{(p)} \gets w_{i, k} - \sum_{j} \mu_k^{(p)}(i, j).{\hspace{3.3cm}}{\ }
\]
\State Once the vector $(w_{j,k+1}^{(i)})_{j\in [n]}$ is known for $i\in [n]$, set
\[
w_{i,k+1} = \left\{\begin{array}{ll}
w_{i,k+1}^{(i)}  & \mbox{ if $w_{i,k+1}^{(i)} > w_{\min} \bydef \frac{\sqrt{n}}{T}$,}\\
0           & \mbox{ otherwise.}
\end{array}\right.
\]
\EndFor
\end{algorithmic}
\end{algorithm}

The \WeightUpdate{} algorithm from the perspective of player $p$ is presented in \autoref{alg:weights-update-real-bb}.
We want to ensure that the fractional matchings computed by good players are numerically very close to each other, 
and for this reason, we use a specific maximal matching
algorithm called \RisingTide{} (\autoref{alg:rising-tide})
that has a continuous Lipschitz property~\cite[Lemma~13]{HuangPZ22}, i.e.,
bounded perturbations to its input yield bounded perturbations to its output.
Other natural maximal matching algorithms such 
as \emph{greedy} do not have this property.

\subsubsection{Rising Tide Algorithm}
\label{sect:rising-tide}

The \RisingTide{} algorithm initializes $\mu=0$ and 
simply simulates
the continuous process of increasing 
all $\mu(i,j)$-values in lockstep, so long as 
$i$, $j$, and $(i,j)$ are not saturated.
At the moment one becomes saturated, $\mu(i,j)$ is 
frozen at its current value.

\begin{algorithm}[h]
\caption{$\RisingTide(H=(V,E,c_V,c_E))$}\label{alg:rising-tide}
\begin{algorithmic}[1]
\State $E'\gets \{(i, j)\in E\ |\ c_E(i, j) > 0\}$.
\State $\mu(i, j)\gets 0$ for all $i,j\in V$.
\While{$E' \neq\emptyset$}
    \State Let $\mu_{E'}(i, j) = \begin{cases} 1 & \text{ if } (i, j)\in E'\\ 0 & \text{ otherwise.}\end{cases}$.
    \State Choose maximum $\epsilon > 0$ such that $\mu'=\mu+\epsilon \mu_{E'}$ is a feasible fractional matching.
    \State Set $\mu\gets \mu'$.
    \State $E' \gets E' - \{(i,j) \mid \mbox{$i$ or $j$ or $(i,j)$ is saturated}\}$\Comment{$\mu(i,j)$ cannot increase}\label{line:remove-edges}
\EndWhile
\State \Return $\mu$.
\end{algorithmic}
\end{algorithm}

Recall that $c_V(i)$ is initialized to be the 
old weight $w_{i,k}$
and 
the new weight in $p$'s local view 
is set to $w_{i,k+1}^{(p)} = c_V(i)-\sum_j \mu_k^{(p)}(i,j)$.  
We are mainly
interested in differences in the new weight 
vector computed by 
players that begin with slightly 
different graphs $H^{(p)},H^{(q)}$.  
\autoref{lem:rising-tide-output}~\cite{HuangPZ22} 
bounds the distance between outputs
in terms of the distance between inputs.

\begin{lemma}[{\cite[Lemma 13]{HuangPZ22}}]\label{lem:rising-tide-output}
Let $H^{(p)}=(V, E, c_V^{(p)}, c_E^{(p)})$ 
and $H^{(q)}=(V, E, c_V^{(q)}, c_E^{(q)})$ 
be two capacitated graphs, 
which differ by $\eta_E = \sum_{i, j} |c_E^{(p)}(i, j) - c_E^{(q)}(i, j)|$ 
in their edge capacities
and 
$\eta_V = \sum_{i} |c_V^{(p)}(i) - c_V^{(q)}(i)|$ 
in their vertex capacities.
Let $\mu^{(p)}$ and $\mu^{(q)}$ be the fractional matchings computed by \RisingTide{} (\autoref{alg:rising-tide}).
Then: 
\[
\sum_{i} \left| \left(c_V^{(p)}(i) -  \sum_j \mu^{(p)}(i, j)\right) - \left(c_V^{(q)}(i) - \sum_j \mu^{(q)}(i, j)\right)\right| \le \eta_V + 2\eta_E.
\]
\end{lemma}

\subsection{Error Accumulation and Reaching Agreement}
\label{sect:bounding-error}

The maximum number of epochs is $\KMAX=3f$.
Let $k\in [1,\KMAX]$ be the index of the current epoch, 
and let $(w_{i, k})$ be the weights that were used in the execution 
of $\CoinFlip$ during epoch $k$.
Upon completing epoch $k$, each player $p$ 
applies \autoref{alg:weights-update-real-bb} 
to update the consensus weight vector $(w_{i, k})_{i\in [n]}$ 
to produce a local weight vector $(w_{i,k+1}^{(p)})_{i\in [n]}$,
and then the consensus weight vector $(w_{i,k+1})_{i\in [n]}$ 
used throughout epoch $k+1$.

\begin{lemma}[Maintaining \autoref{inv:weights}]\label{lem:maintaining-inv}
Suppose for some $\epsilon>0$ that $n=(3+\epsilon)f$, 
$m=\Omega(n\ln n/\epsilon^4)$, 
and $T=\Omega(n^2\ln^3 n/\epsilon^4)$.
At any point in epoch $k\in[1,\KMAX]$, with high probability,
\[
\sum_{i\in G} (1 - w_{i, k}) \le \sum_{i\in B} (1 - w_{i, k}) + \frac{\epsilon^4}{\sqrt{n\ln^6 n}}\cdot (k-1).
\]
\end{lemma}

\begin{proof}
By induction on $k$.
For the base case $k=1$ all the weights are $1$ so the claim clearly holds. 
Now suppose the claim holds for $k$ and consider $k+1$.
Fix any good player $p$. 
A consequence of Part 1 of \autoref{lem:n/3-gap-lemma}
is that with high probability, 
player $p$'s view of the weight vector, 
$(w_{i,k+1}^{(p)})$, 
is derived from $(w_{i,k})$ 
by deducting at least as much weight from bad players as from good players. By the inductive hypothesis,
\begin{align*}
    \sum_{i\in G} (1 - w_{i, k+1}^{(p)}) \le \sum_{i\in B} (1 - w_{i, k+1}^{(p)}) + \frac{\epsilon^4}{\sqrt{n\ln^6 n}}\cdot (k-1). 
\end{align*}
Subsequently, player $p$ derives the consensus 
weight vector $(w_{i,k+1})$ from $(w_{i,k+1}^{(q)})_{q\in [n], i\in [n]}$ by setting $w_{i,k+1}=w_{i,k+1}^{(i)}$, rounding the value down to 0 if it is at most $w_{\min}$. Therefore, 
\begin{align*}
\sum_{i\in G} (1 - w_{i, k+1}) &\le \sum_{i\in B} (1 - w_{i, k+1}) + \frac{\epsilon^4}{\sqrt{n\ln^6 n}}\cdot (k-1) + \sum_{i\in [n]} |w_{i, k+1}^{(p)} - w_{i, k+1}^{(i)}| + w_{\min{}} n_0,
\end{align*}
where $n_0$ is the number of players
whose weight is rounded 
down to $0$ after epoch $k$.

Hence, it suffices to show that 
$\sum_{i\in [n]} |w_{i, k}^{(p)} - w_{i, k}^{(i)}| + w_{\min{}}n_0\le \epsilon^4/\sqrt{n\ln^6 n}$.
By \autoref{lem:rising-tide-output}, the computed weight difference 
between player $p$ and any player $q$ 
can be bounded by twice the sum of all edge capacity differences ($\eta_E$),
since they agree on the vertex capacities ($\eta_V=0$).\footnote{Strictly speaking one player can know $c_V(i)$ and another player 
may be ignorant of it, but in this case all the edges incident to $i$ have capacity zero. This situation is indistinguishable from all players knowing and agreeing on $c_V$.}
According to \autoref{alg:weights-update-real-bb}, 
the edge capacities differ 
due to underlying disagreement on the $\CORR(i, j)$ values. Thus,

\begin{align}
|w_{q, k+1}^{(p)} - w_{q, k+1}^{(q)}| 
&\leq 2\eta_E
\leq 2\cdot \frac{8}{\epsilon^2 fmT}\sum_{i\neq j} \left| \CORR^{(p)}(i, j) - \CORR^{(q)}(i, j) \right|.
\label{eqn:etaE}
\end{align}
By \autoref{thm:blackboard}, two players may only disagree in up to $f$ cells of the blackboards
$$(\BB_{2(k-1)T+2},\BB_{2(k-1)T+4},\ldots,\BB_{2kT}), $$ i.e.,
those used to compute $\CORR$-values in epoch $k$.
Since the sum of each column in 
each blackboard is bounded by $\XMAX$, 
$|X_i^{(p)}(t)X_j^{(p)}(t) - X_i^{(q)}(t)X_j^{(q)}(t)|\leq 2\XMAX$.  
Each of the $f$ cells that $p$ and $q$ disagree about 
affects $n-1$ $\CORR$-values.
Therefore, 
the right hand side of (\ref{eqn:etaE}) is upper 
bounded by:
\begin{align*}
&\le 2\cdot \frac{8}{\epsilon^2 fmT}\cdot nf\cdot 2\XMAX
\\
&\le \frac{32 n\XMAX}{\epsilon^2 mT} \\
&\le \frac{ \sqrt{n}}{T}\tag{$m=\Omega(n\ln n/\epsilon^4)$ and $\XMAX = \Theta(\sqrt{m\ln n})$}\\
&=w_{\min{}}
\end{align*}
Now the inductive step for $k+1$ holds by noticing that
\begin{align*}
\sum_{i\in [n]}|w_{i, k+1}^{(p)} - w_{i, k+1}^{(i)}|  + w_{\min{}}n_0 &\le  2w_{\min{}}n\\
&\le \frac{\epsilon^4}{n^{3/2}\ln^{3} n}\cdot n  \tag{$T=\Omega(n^2\ln^3 n/\epsilon^4)$}\\
&= \frac{\epsilon^4}{\sqrt{n\ln^6 n}}.
\end{align*}
Since $k\leq \KMAX = 3f$, we conclude that \autoref{inv:weights} holds in every epoch, with high probability. 
That is, if $(w_i),G,B$ are the weight vector, good players, and bad players  at any point in time, then
\begin{align*}
\sum_{i\in G}(1-w_i) &\le \sum_{i\in B}(1 - w_i) + \frac{\epsilon^4}{\sqrt{n\ln^6 n}} \cdot \KMAX\\
&\le \sum_{i\in B}(1 - w_i) 
+ \epsilon^4 f  \tag{$\sqrt{n\ln^6 n} \geq 3$}.
\end{align*}
\end{proof}

The next observation and \autoref{lem:weight-update-progress} shows that the weight 
of every bad player becomes $0$ after running $\KMAX$ epochs of \WeightUpdate{}s 
without reaching agreement.

\begin{observation}\label{obs:weights-closed-to-zero}
For any $i$ and $k$, if there exists a player 
$p$ such that $w_{i, k}^{(p)}=0$, then $w_{i, k} = 0$.
\end{observation}

\begin{proof}
In the proof of \autoref{lem:maintaining-inv} it was shown that 
$|w_{i, k}^{(p)} - w_{i, k}^{(i)}|\le \sqrt{n}/T = w_{\min{}}$,
hence if $w_{i,k}^{(p)}=0$, $w_{i,k}^{(i)}\leq w_{\min{}}$ and 
$w_{i,k}$ is rounded down to 0. 
See \autoref{alg:weights-update-real-bb}.
\end{proof}

\begin{lemma}
\label{lem:weight-update-progress}
If agreement has not been reached after $\KMAX=3f$ epochs, then with high probability,
there are $f$ bad 
players with weight $0$.
\end{lemma}

\begin{proof}
There are at most $f$ epochs in which the adversary corrupts at 
least one player.
We argue below that for all other epochs, in the call to 
\WeightUpdate, the sum of the capacities of edges with at least one endpoint in $B$ is at least 1.
This implies that in each iteration of \WeightUpdate,
either some $i\in B$ with $c_V(i)=w_i>w_{\min{}}$ becomes
saturated (and thereafter $w_i=0$ by \autoref{obs:weights-closed-to-zero}), 
or the total weight of all players in $B$ 
drops by at least 1. 
Each of these cases also occurs at most $f$ times,
hence $\KMAX=3f$ epochs suffice to push the
weight of $f$ bad players to zero, with high probability.

We now prove that the sum of the capacities of edges with at least one endpoint in $B$ is at least 1. 
\begin{align*}
\sum_{(i,j) \in [n] \times B} c_E(i, j)
    &= \frac{8}{\epsilon^2 fm T} \sum_{(i,j)\in [n] \times B} \max\{0, \CORR(i, j) - w_{i, k}w_{j, k}\beta\}  \\
    &\geq \frac{8}{\epsilon^2 fm T} \sum_{(i,j)\in G \times B} \max\{0, \CORR(i, j) - w_{i, k}w_{j, k}\beta\}  \\
    &\ge \frac{8}{\epsilon^2 fmT} \left( \frac{1}{8}\epsilon^2 fmT\right)  \tag{by \autoref{lem:n/3-gap-lemma}(2)}\\
    &= 1.\qedhere
\end{align*}
\end{proof}

\begin{lemma}\label{lem:good-processes-agree}
Suppose \autoref{inv:weights} holds in an 
epoch in which bad players have zero weight.
With high probability, \BrachaAgreement{} 
terminates with agreement in this epoch.
\end{lemma}

\begin{proof}
The proof of \autoref{lem:n/3-gap-lemma}(2) 
states that with high probability, 
$-\sum_{(i,j)\in G\times B} \CORR(i,j) \geq \frac{1}{4}\epsilon^2 fmT >0$
in any epoch in which \BrachaAgreement{} fails to reach agreement.
On the other hand, by assumption $-\sum_{(i,j)\in G\times B} \CORR(i,j)=0$.
Thus, with high probability \BrachaAgreement{} reaches agreement 
in this epoch.
\end{proof}

\begin{remark}
Assuming the preconditions of \autoref{lem:good-processes-agree},
only good players participate and flip fair coins.
Nonetheless, the output of \CoinFlip{} 
need not be close to unbiased.
The scheduling power of the adversary is still strong 
enough to
heavily bias the outcome of the coin flip or create disagreement on the outcome.
More careful calculations show the probability of an unambiguous $-1$ or unambiguous $+1$ outcome are 
each $\Omega(1/\sqrt{n})$.  This occurs when
$f+1$ good players have total weight about $1$ 
and the remaining $f$ good players each have weight 1.
\end{remark}

\begin{theorem}
Suppose $n=(3+\epsilon)f$ where $\epsilon>0$, 
$m=\Theta(n\ln n/\epsilon^4)$, and 
$T=\Theta(n^2\ln^3 n/\epsilon^4)$.
Using the implementation of $\CoinFlip$ from \autoref{sect:implementation-of-coinflip},
\BrachaAgreement{} solves Byzantine agreement 
with probability 1 in the full information, 
asynchronous model against an adaptive adversary. 
In expectation the total latency 
is $\tilde{O}(n^4/\epsilon^8)$.
The local computation of each player 
is polynomial in $n$.
\end{theorem}

\begin{proof}
By \autoref{lem:weight-update-progress}, 
if the players have failed to reach agreement after $\KMAX=3f$ epochs, 
all bad players' weights are zero with high probability.
By \autoref{lem:good-processes-agree}, 
the players reach agreement in the next 
epoch with high probability.
If, by chance, the players have not reached 
agreement after epoch $\KMAX+1$, 
they reset all weights to 1 and restart 
the algorithm at epoch 1.
Thus, the algorithm terminates with probability 1.

By \autoref{thm:blackboard} the 
latency to construct $\BB_{t'}$ is $O(m(t'))$, 
so if the algorithm is not restarted, 
the total latency is 
$O((\KMAX+1) m T) = O(fmT) = \tilde{O}(n^4/\epsilon^8)$.
Even after incorporating the possibility of restarts, 
this latency bound holds
with high probability and in expectation.
\end{proof}

\section{Conclusion}\label{sect:conclusion}

Our main result is a randomized agreement protocol in the full-information model
resilient to $f = n/(3+\epsilon)$ adaptive Byzantine failures.
When $\epsilon$ is bounded away from zero it
has expected latency $\tilde{O}(n^4)$ 
and in the extreme case when 
$\epsilon=1/f$ ($n=3f+1$)
it has latency
$\tilde{O}(n^{12})$.
This is the first improvement to Bracha's~\cite{Bracha1987} 
1984 protocol when $n=3f+1$,
which had exponential expected latency $2^{O(n)}$.
The prior results of King and Saia~\cite{KingS2016,KingS2018}
and Huang, Pettie, and Zhu~\cite{HuangPZ22} 
offer different tradeoffs between resiliency, 
latency, and local computation time; see Table~\ref{tab:BA}.

Our goal was to present the simplest possible protocol with polynomial latency and resiliency $f<n/3$, so we have made no 
attempt to optimize the latency 
prematurely.\footnote{For example, there is a mismatch in the two terms 
of 
\autoref{lem:sumSGsquared} ($\frac{1}{4}\epsilon^2 mfT - mn\sqrt{T(c\ln n)}$), 
the first achieves its worst case in later epochs, after close to $2f$ units of weight have been deducted from $(w_i)$,
and the second achieves its worst case in the beginning, when $w_i=1$.  Similarly \autoref{lem:biasSG} achieves its 
worst case when $w_i=1$.  These mismatches lead to the introduction of additional $\epsilon^{-1}$ factors
in the proof of \autoref{lem:n/3-gap-lemma}(2).}
We see several interesting directions for future work.
\begin{itemize}
\item
If one tried to stay broadly within our coin-flipping framework
but maximally optimize the parameters, one would find that 
$m=\Omega(n)$ and $T=\Omega(n^2)$.  The bound on $m$ is necessary to overcome the $\pm f$ discrepancies introduced by the adversarial scheduler,
and the bound on $T$ is the time needed for the small negative 
aggregate good-bad correlations to dwarf the random noise of good-good correlations.  
Thus, $\Omega(n^3)$ latency seems to be baked into this approach when $n=3f+1$, even just to reliably blacklist \emph{any} pair of players
and maintain \autoref{inv:weights}.
Is it possible to solve Byzantine Agreement 
in \emph{close} to $O(n^3)$ latency?
\item Let 
$n^{g(\epsilon)}$ be the optimum latency 
of a Byzantine Agreement protocol 
with resiliency $f=n/(3+\epsilon)$ 
against an adaptive adversary.
What does the function $g$ look like 
and what is 
$\lim_{\epsilon\rightarrow \infty} g(\epsilon)$?  
From~\cite{KingS2016} we know that  
$\lim_{\epsilon\rightarrow \infty} g(\epsilon) \leq 2.5$, 
at least for protocols with exponential local computation,
and from~\cite{Aspnes98,attiya2008tight} we know $\lim_{\epsilon\rightarrow \infty} g(\epsilon) \geq 1$.
What is the correct limit of $g(\epsilon)$?  
Are there qualitatively different
protocols achieving latency $n^{g(\epsilon)}$ for various ranges of $\epsilon$?  One can also look at the 
optimal latency-resiliency tradeoff when 
$f=n^\gamma$, $\gamma\in (0,1)$, 
is expressed as a polynomial of $n$.
\item Each step in the protocols we use (Bracha's~\cite{Bracha1987} \ReliableBroadcast, Huang et al.'s~\cite{HuangPZ22} \IteratedBlackboard, and the \CoinFlip{} protocol 
of \autoref{sect:implementation-of-coinflip})
typically consists of sending a message to all players
and waiting for $n-f$ messages before making some state transition.
If we were to wait for $n-f+1$ messages, we may wait forever if 
$f$ players crashed and never sent any messages.
On the other hand, once we introduce \emph{blacklisting} 
it is not clear that waiting for just 
$n-f$ messages is necessary anymore.
For example, 
suppose that $\sum_i w_i = n - 2\rho f$ and that we have reduced
the weight of good and bad 
players each by $\rho f$, with high probability.
Rather than wait for $n-f$ messages, 
we could wait for 
messages from players whose 
\emph{total weight} 
is at least $n-(\rho+1)f$.\footnote{More generally, 
players should make a state transition only when it 
is \emph{plausible, w.h.p.}, that all un-received messages
were to be sent by Byzantine players.}
This would help speed up later epochs since we could
then access the weight advantage of 
good players.
However, since there is \emph{some} non-zero probability of blacklisting 
pairs of good players, there is \emph{some} 
non-zero probability that a protocol 
will deadlock
if it waits for $n-(\rho+1) f$ 
weight before proceeding.
This raises the possibility 
that there is a complexity separation between
\emph{Las Vegas} (which never err) 
and \emph{Monte Carlo} protocols
(which may deadlock) for Byzantine Agreement in our model. Cf.~\cite{KapronKKSS10}.
\end{itemize}


\newcommand{\etalchar}[1]{$^{#1}$}

\end{document}